\documentclass[11pt]{article}
\usepackage[T1]{fontenc}
\usepackage[utf8]{inputenc}
\usepackage{mathtools}
\usepackage[colorlinks, citecolor=blue]{hyperref}
\usepackage{makecell}
\usepackage{thm-restate}
\usepackage{xcolor}
\usepackage{tikz}
\usetikzlibrary{arrows}
\usetikzlibrary{calc,patterns,backgrounds}
\usetikzlibrary{fit}

\title{Robust Testing and Estimation under Manipulation Attacks}

\author{Jayadev Acharya\thanks{This research is supported by the NSF through grants CCF-1846300 (CAREER), CCF-1815893.}\\
Cornell University\\
\tt{acharya@cornell.edu}
\and
Ziteng Sun\\
Cornell University\\
\tt{zs335@cornell.edu}
\and
Huanyu Zhang\\
Cornell University\\
\tt{hz388@cornell.edu}
}

\usepackage{fullpage}
\usepackage{amsfonts,amsmath,amssymb,amsthm, pbox}
\usepackage{mleftright}

\usepackage{multirow}

\usepackage{dsfont} %

\usepackage[shortlabels]{enumitem}
\setitemize{noitemsep,topsep=0pt,parsep=0pt,partopsep=0pt}
\setenumerate{noitemsep,topsep=0pt,parsep=0pt,partopsep=0pt}

\makeatletter
\@ifundefined{theorem}{%
  \theoremstyle{definition}
  \newtheorem{definition}{Definition}
  \theoremstyle{plain}
  \newtheorem{theorem}{Theorem}
  \newtheorem{corollary}[theorem]{Corollary}
  \newtheorem{lemma}[theorem]{Lemma}

  \theoremstyle{remark}
  \newtheorem{remark}{Remark}

}{}
\makeatother

\newcommand{\ignore}[1]{}

\newcommand{\EE}{\mathbb{E}}

\newcommand{\RR}{\mathbb{R}}

\newcommand{\expectation}[1]{\EE\left[#1\right]}

\def \cP     {{\cal P}}

\def \cT     {{\cal T}}

\def \cW     {{\cal W}}
\def \cX     {{\cal X}}
\def \cY     {{\cal Y}}

\def \ceil#1{{\lceil{#1}\rceil}}

\newcommand{\absv}[1]{\left|#1\right|}

\def \Paren#1{{\left({#1}\right)}}

\newcommand{\ed}{\stackrel{\text{def}}{=}}
\newcommand{\eqdef}{{:=}}

\newcommand{\probof}[1]{\Pr\Paren{#1}}

\def\ignore#1{}

\newcommand{\bi}{\begin{itemize}}
\newcommand{\ei}{\end{itemize}}

\def\orpro{\mathop{\mathchoice
   {\vee\kern-.49em\raise.7ex\hbox{$\cdot$}\kern.4em}
   {\vee\kern-.45em\raise.63ex\hbox{$\cdot$}\kern.2em}
   {\vee\kern-.4em\raise.3ex\hbox{$\cdot$}\kern.1em}
   {\vee\kern-.35em\raise2.2ex\hbox{$\cdot$}\kern.1em}}\limits}

\def\andpro{\mathop{\mathchoice
 {\wedge\kern-.46em\lower.69ex\hbox{$\cdot$}\kern.3em}
 {\wedge\kern-.46em\lower.58ex\hbox{$\cdot$}\kern.25em}
 {\wedge\kern-.38em\lower.5ex\hbox{$\cdot$}\kern.1em}
 {\wedge\kern-.3em\lower.5ex\hbox{$\cdot$}\kern.1em}}\limits}

\def\simge{\mathrel{%
   \rlap{\raise 0.511ex \hbox{$>$}}{\lower 0.511ex \hbox{$\sim$}}}}

\def\simle{\mathrel{
   \rlap{\raise 0.511ex \hbox{$<$}}{\lower 0.511ex \hbox{$\sim$}}}}

\newcommand{\ab}{k}
\newcommand{\ns}{n}
\newcommand{\priv}{\eps}
\newcommand{\dist}{\alpha}
\newcommand{\risk}{R}
\newcommand{\setab}{\triangle_\ab}

\newcommand{\vecz}{\textbf{z}}
\newcommand{\batch}{B}
\newcommand{\nb}{N}
\newcommand{\no}{{\texttt{no} }}
\newcommand{\yes}{{\texttt{yes} }}

\newcommand{\ham}[2]{d_{\rm Ham}(#1,#2)}
\newcommand{\dtv}[2]{d_{\rm TV}(#1,#2)}

\newcommand{\norm}[2]{\left\lVert#1\right\rVert_{#2}}
\newcommand{\chisquare}[2]{d_{\chi^2}(#1,#2)}

\newcommand{\p}{p}
\newcommand{\q}{q}
\newcommand{\unif}{u}
\newcommand{\mdist}{Q}
\newcommand{\dem}[2]{{d_{ \rm EM}\left(#1, #2\right)}}
\newcommand{\idc}[1]{\mathbf{1}\left\{#1\right\}}
\newcommand{\dl}{{\rm DL}}
\newcommand{\dit}{{\rm IT}}
\newcommand{\dut}{{\rm UT}}

\newcommand{\hp}{\widehat{\p}}
\newcommand{\tp}{\tilde{\p}}
\newcommand{\rp}{\tp}

\newcommand{\Xon}{X^\ns}
\newcommand{\Yon}{Y^\ns}
\newcommand{\Zon}{Z^\ns}

\newcommand{\bcube}{\{\pm 1\}^{k/2}}

\newcommand{\eps}{\varepsilon}

\newcommand{\expectover}[2]{\EE_{#1}\left[{#2}\right]}
\newcommand{\probover}[2]{\Pr_{#1}\left[{#2}\right]}

\newcommand{\condprob}[2]{\Pr\left({#1}\;\middle|\;{#2}\right)}

\newcommand{\numc}{m}
\newcommand{\fracc}{\gamma}

\newcommand{\inp}{X}
\newcommand{\iter}{i}

\newcommand{\out}{Y}
\newcommand{\outron}{\out^\ns}

\newcommand{\per}{Z}
\newcommand{\perr}{z}
\newcommand{\perron}{\Zon}

\newcommand{\uni}{u[\ab]}

\newcommand{\Ue}{U}

\newcommand{\mapping}{\varphi}

\newcommand{\smbb}[1]{M_{\perr}(#1) }
\newcommand{\uniff}[1]{u[#1]}

\newcommand{\stat}[1]{S(#1)}

\newcommand{\corr}{\gamma}
\newcommand{\bernoulli}[1]{{\rm Ber}(#1)}
\newcommand{\binomial}[2]{{\rm Binom}(#1, #2)}

\newcommand{\newer}[1]{\textcolor{black}{#1}}
\newcommand{\hz}[1]{{\color{black}#1}} %
\newcommand{\hzz}[1]{{\color{black}#1}} %
\newcommand{\zs}[1]{{\color{black}#1}} %
\newcommand{\zsnew}[1]{{\color{black}#1}} %

\makeatletter
\def\thanks#1{\protected@xdef\@thanks{\@thanks
        \protect\footnotetext{#1}}}
\makeatother

\begin{document}

\maketitle

\begin{abstract}
We study robust testing and estimation of discrete distributions in the 
strong contamination model. \newer{We consider} both the ``centralized setting'' and 
the ``distributed setting with information constraints''
including communication and local privacy (LDP) constraints.

Our technique relates the strength of manipulation attacks to the 
earth-mover distance using Hamming distance as the metric between 
messages (samples) from the users.
In the centralized setting, we provide 
optimal error bounds for both learning and testing. Our lower bounds under 
local information constraints build on the recent lower bound methods 
in 
distributed 
inference. In the communication 
constrained 
setting, we 
develop novel 
algorithms based on random hashing and an $\ell_1/\ell_1$ 
isometry.
\end{abstract}

\section{Introduction}
\hz{Data from users form the backbone of modern distributed learning 
systems such as federated learning~\cite{kairouz2019advances}.
}
Two of the key aspects of such \hz{large-scale} distributed systems that 
make 
inference tasks challenging are 
\begin{enumerate}[label=(\roman*)] 
\item \textit{information constraints} at the users (e.g., 
preserving privacy, bandwidth limitations), and
\item \textit{$\corr$-manipulation attacks} where an 
adversary has complete control over a 
$\corr$ fraction of the users.
\end{enumerate}
\newer{An extreme example is when  
malicious users are deliberately injected to disrupt the system. Note that when there are only manipulation attacks but no information constraints, the setting is equivalent to the robust inference where a fraction of the samples can be adversarially corrupted.} 

\cite{cheu2019manipulation} initiated the study of manipulation attacks under  local differential privacy (LDP), thereby considering the practically important setting where both of the challenges above  \hz{exist} 
simultaneously. Motivated by their work, \hz{we} further study 
manipulation attacks for inference on discrete distributions both with and 
without information constraints. 

\begin{figure}[h]\centering
	\scalebox{.9}{\begin{tikzpicture}[->,>=stealth',shorten >=1pt,auto,node distance=20mm, semithick]
  \node[circle,draw,minimum size=13mm] (A) {$X_1$};
  \node[circle,draw,minimum size=13mm] (B) [right of=A] {$X_2$};
  \node[circle,draw,minimum size=13mm] (BB) [right of=B] {$X_3$};
  \node (C) [right of=BB] {$\dots$};
  \node[circle,draw,minimum size=13mm] (DD) [right of=C] {$X_{\ns-2}$};
  \node[circle,draw,minimum size=13mm] (D) [right of=DD] {$X_{\ns-1}$};
  \node[circle,draw,minimum size=13mm] (E) [right of=D] {$X_\ns$};
  
  \node[rectangle,draw,minimum width=13mm,minimum 
  height=7mm
  ,fill = white] (WA) [below of=A, node distance = 
  15mm] 
  {$W_1$};
  \node[rectangle,draw,minimum width=13mm,minimum 
  height=7mm%
  ,fill = white] (WB) [below of=B, node distance = 
  15mm] {$W_2$};
  \node[rectangle,draw,minimum width=13mm,minimum 
  height=7mm%
  ,fill = white] (WBB) [below of=BB, node distance = 
  15mm] {$W_3$};
  \node[rectangle,minimum width=13mm,minimum height=7mm,fill=none] 
  (WC) [below of=C, node distance = 
  15mm] {$\dots$};
  \node[rectangle,draw,minimum width=13mm,minimum 
  height=7mm%
  ,fill = white] (WDD) [below of=DD, node distance = 
  15mm] {$W_{\ns-1}$};
  \node[rectangle,draw,minimum width=13mm,minimum 
  height=7mm%
  ,fill = white] (WD) [below of=D, node distance = 
  15mm] {$W_{\ns-1}$};
  \node[rectangle,draw,minimum width=13mm,minimum 
  height=7mm%
  ,fill = white] (WE) [below of=E, node distance = 
  15mm] {$W_\ns$};
  
  \node[circle,draw,minimum size=13mm,fill=white] (YA) [below of=WA, 
  node distance = 
  15mm] 
  {$Y_1$};
  \node[circle,draw,minimum size=13mm,fill=white
  ] (YB) [below of=WB, node distance = 
  15mm] {$Y_2$};
  \node[circle,draw,minimum size=13mm,fill=white] (YBB) [below of=WBB, 
  node distance = 
  15mm] {$Y_3$};
  \node[circle,minimum size=13mm] (YC) [below of=WC, node distance = 
  15mm] {$\dots$};
  \node[circle,draw,minimum size=13mm,fill=white
  ] (YDD) [below of=WDD, node distance = 
  15mm] {$Y_{\ns-2}$};
  \node[circle,draw,minimum size=13mm,fill=white] (YD) [below of=WD, 
  node distance = 
  15mm] {$Y_{\ns-1}$};
  \node[circle,draw,minimum size=13mm,fill=white] (YE) [below of=WE, 
  node distance = 
  15mm] 
  {$Y_\ns$};
  
  \node[circle,draw,minimum size=13mm,fill=white] (MA) [below of=YA] 
  {$Y_1$};
  \node[circle,draw,minimum size=13mm,fill=white
  ,fill=gray!20!white] (MB) [below of=YB] {$Y'_2$};
  \node[circle,draw,minimum size=13mm,fill=white] (MBB) [below of=YBB] 
  {$Y_3$};
  \node[circle,minimum size=13mm] (MC) [below of=YC] {$\dots$};
  \node[circle,draw,minimum size=13mm,fill=white
  ,fill=gray!20!white] (MDD) [below of=YDD] {$Y'_{\ns-2}$};
  \node[circle,draw,minimum size=13mm,fill=white] (MD) [below of=YD] 
  {$Y_{\ns-1}$};
  \node[circle,draw,minimum size=13mm,fill=white] (ME) [below of=YE] 
  {$Y_\ns$};
  \node[draw,dashed,fit=(MA) (MB) (MBB) (MC) (MDD) (MD) (ME)] {};
  
  \node (P) [above of=C] {$\p$};
  \node[rectangle,draw, minimum size=10mm, node distance = 25mm] (R) 
  [below of=MC, node distance = 20mm] {Server};
  \node (out) [below of=R,node distance=13mm] {output};
	\node
	(AD) 
	[above of=MC, node distance = 4mm] {Adversary};
	
  \draw[->] (P) edge[densely dashed,bend right=10] (A)(A) edge (WA)(WA) 
  edge (YA)(YA) edge (MA)(MA)edge[bend right=10] (R);
  \draw[->] (P) edge[densely dashed,bend right=5] (B)(B) edge (WB)(WB) edge 
  (YB)(YB) edge (MB)(MB) edge[bend right=5] (R);
  \draw[->] (P) edge[densely dashed] (BB)(BB) edge (WBB)(WBB) edge (YBB)(YBB) 
  edge (MBB)(MBB) edge (R);
  \draw[->] (P) edge[densely dashed] (DD)(DD) edge (WDD)(WDD) edge 
  (YDD)(YDD) edge (MDD)(MDD)  edge (R);
  \draw[->] (P) edge[densely dashed,bend left=5] (D)(D) edge (WD)(WD) edge 
  (YD)(YD) edge (MD)(MD) edge[bend left=5] (R);
  \draw[->] (P) edge[densely dashed,bend left=10] (E)(E) edge (WE)(WE) edge 
  (YE)(YE) edge (ME)(ME) edge[bend left=10] (R);
  \draw[->] (R) edge (out);
\end{tikzpicture}}
		\caption{The information-constrained distributed model with 
		manipulation attack. The shaded messages (at most $\fracc \ns$ 
		messages) 
		are manipulated by the adversary.} 
	\label{fig:model}
\end{figure}
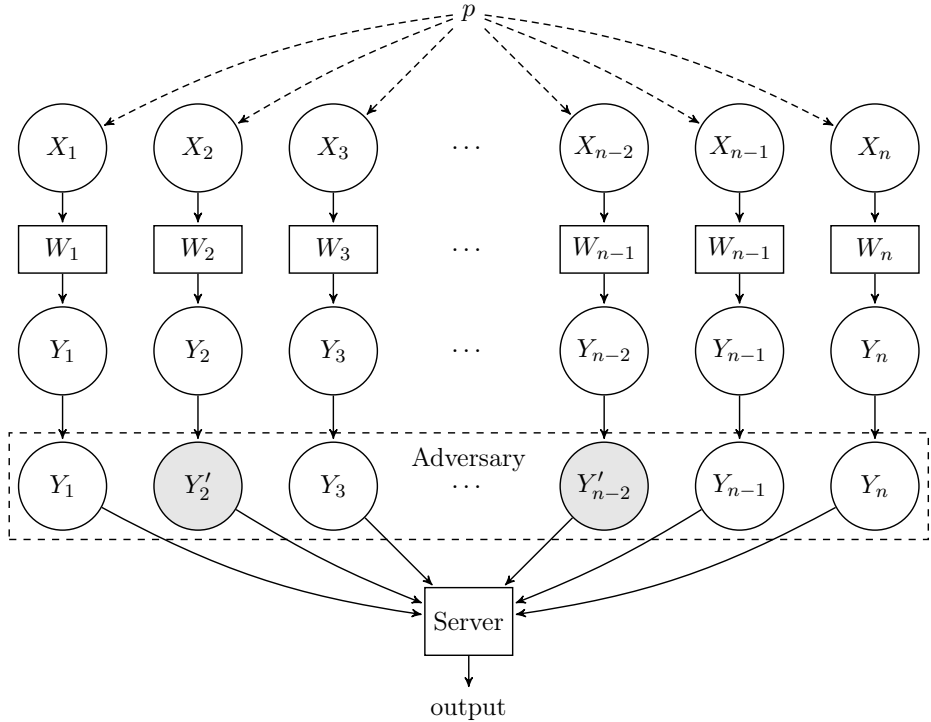
\medskip
\noindent\textbf{Problem setup.} \newer{Let $\setab$ be the simplex of all distributions over a discrete domain 
$\cX$ of size $\ab$} \newer{(wlog let $\cX=[\ab]:=\{1,\ldots, \ab\}$), and
 $X^\ns:=(X_1, \ldots, X_\ns)$ be $\ns$ independent samples from an unknown $\p\in\setab$ which
are distributed across $\ns$ users.} Each user $i$ then sends a message $Y_i 
\in \cY$ 
based on 
$X_i$ according to a pre-specified communication protocol $\Pi$ to the 
server.

\newer{An adversary has access to $\Pi$ and observes the \emph{intended} messages $Y^\ns:= (Y_1,\ldots, Y_\ns)$. It \emph{then} chooses a set 
$C\subset [\ns]$ with $|C|\le \numc \eqdef \fracc \ns$ and performs an attack $M_C: 
\cY^\ns \rightarrow \cY^\ns$ as follows: for each $i\in C$, it can change $Y_i$ to an arbitrary $Y_i'$. The output of this attack is $Z^\ns \eqdef (Z_1, \ldots, Z_n) = M_C(Y^n)$, which satisfies $Z_i=Y_i$ if $i\notin C$, 
and $Z_i=Y_i'$ if $i\in C$. We call this a \emph{$\corr$-manipulation attack}. A central server observes $\Zon$ (and has no knowledge about $C$ or $M_C$) and has to 
solve an inference task on $\p$. See Figure~\ref{fig:model} for an overview 
of the model.}

\begin{remark}
\newer{A natural question to ask is what happens if the adversary, in addition to $\Yon$, can also observe original samples $\Xon$. The algorithms proposed in this paper all work with the same guarantee under this setting as well. 
Moreover, the proposed attacks only use $\Yon$, and therefore any optimality result naturally holds under this stronger threat model as well.}
\end{remark}

\begin{remark}
\newer{The threat model is closer to 
the strong contamination 
	model~\cite{diakonikolas2019recent} considered in recent 
	literature 
	on robust statistics. This adversary is stronger than the setting of
	~\cite{cheu2019manipulation} where it is only allowed to 
	select $C$ and choose $Y_i'$s based on the protocol $\Pi$ instead of the 
	messages $\Yon$.}
\end{remark}

\medskip
\noindent\textbf{Communication protocol.} We consider public-coin non-interactive 
protocols in this work. The users have access to public 
randomness $U$, which is independent of $X^\ns$. Based on $\Ue$, user 
$i$ chooses a 
channel $W_i$, which is a (possibly randomized) mapping described by
\begin{align}
\label{eqn:channel}
W_i(y \mid x) = \condprob{Y_i = y}{X_i = x}.
\end{align}
\newer{When the input distribution is $\p$ and the channel is $W_i$ the distribution of $Y_i$ is 
\begin{align}
\label{eqn:channel-output}
\probof{{Y_i = y}}= \sum_x \p(x) W_i(y|x) = \EE_{X\sim\p}\left[W_i(y|X) \right].
\end{align}
For a given set of channels $W^n$ and input distribution $\p$, let $p^{\Yon}$ denote the output distribution of messages, and by independence of $X_i$'s,
\begin{align}
\label{eqn:channel-all}
\p^{\Yon}(y^n) = \prod_{i=1}^{n} \EE_{X\sim\p}\left[W_i(y_i|X) \right]
\end{align}}
All users then send their messages $Y_i = W_i(X_i), i \in [\ns]$ to the server 
simultaneously. \newer{The adversary also observes $\Ue$ and can use this 
information in its attack (e.g., choose $C$ dependent on $\Ue$ as well).}

\medskip
\noindent\textbf{Information constraints.} 
We model information constraints at the users by a set of channels 
$\cW$ with input domain $[\ab]$. We illustrate this with two canonical 
examples, local differential privacy (LDP) and communication constraints. 

\medskip
\noindent\textit{Local differential privacy (LDP).} A channel 
$W:[\ab]\to\cY=\{0,1\}^\ast$ is $\priv$-LDP if
\begin{align}
\sup_{y \in \cY} \sup_{x, x' \in \cX} \frac{W(y \mid x)}{W(y \mid x')} \le 
e^{\priv},\label{eqn:ldp}
\end{align}
\hz{which requires that the output distributions are close no matter what 
	the 
	input is}, \zsnew{hence protecting the identity of the input.} We use 
	$\cW_{\priv}$ to denote the set of all 
$\priv$-LDP 
channels.

\medskip
\noindent\textit{Communication constraints.} Let $\ell<\log \ab$, and  
$\cW_\ell :=\{W \colon [\ab]\to\cY=\{0,1\}^\ell\}$ be the set of channels 
that output $\ell$-bit messages.

\medskip
\noindent\textbf{Inference tasks.} 
We consider the fundamental tasks of distribution estimation (learning) and goodness-of-fit (identity testing), described below.  

\medskip
\noindent\textit{Distribution learning (DL).} The goal is to design messaging schemes and an estimator $\hp: 
\cY^\ns \rightarrow \setab$ for the underlying distribution $\p$. The 
loss is measured in the expected total variation (TV) distance between 
$\hp$ 
and $\p$, i.e.,
$
\expectover{\p}{\dtv{\hp(Z^n)}{\p}},
$
where the expectation is over the randomness of samples $X^n\sim \p$ and the scheme.
We wish to characterize the following minimax loss (risk) under 
manipulation attacks, where we design the best messaging schemes 
$W^n:=(W_1, \ldots W_n)\in\cW^n$ and estimator $\hat p$ for the worst 
distribution\footnote{Note here by definition, $\sup_{M_C: |C|\le \fracc 
\ns }$ should be inside the expectation since the attacker can observe the 
messages. However, since $M_C$ is a function of $Y^n$, $\sup_{M_C}$ already covers all possible attacks. Hence both 
minimax formulations have the same quantity.}:
\[
	\risk_{\dl}(\ab, \ns, \cW, \fracc) \!:= \! \inf_{\hat{\p}, W^n} 
	\sup_{\p}  \! \sup_{M_C: |C|\le \fracc \ns }  \! 
	 \! \expectation{\dtv{\widehat{p}(Z^n)}{\p}}.
\]

Without any information constraints and without any manipulation 
(i.e., $\fracc = 0$), when the server observes $\Xon$, the risk is known to 
be $\Theta(\sqrt{k/n})$ achieved by the empirical histogram.

\medskip
\noindent\textit{Identity testing (IT).} \newer{Let $\q\in\setab$ be a \emph{known} reference distribution and $\dist>0$ be a distance parameter.
The goal is to design $\cW^n$ and a tester $\cT:\cY^\ns\to\{\yes,\no\}$ 
such that under 
any $\corr$-manipulation attack $M_C$,}
\begin{equation}
\begin{aligned}
&\probover{\p}{\cT(\Zon) = 
		\text{\yes}}>0.9, &&\text{ if $\p=\q$,}\\
&\probover{\p}{\cT(\Zon) = 
		\text{\no}}>0.9, &&\text{ if $\dtv{\p}{\q}\ge\dist$}.
\end{aligned}
\label{eqn:cases-testing}
\end{equation}
In other words, with probability at least 0.9, we can test if $\p=\q$ or $\p$ 
is $\dist$-far in total variation distance from $\q$.
The minimax risk of identity testing under
manipulation attacks is
\begin{align*}
\risk_{\dit}(\ab, \ns, \cW, \corr) := \inf\{ \alpha \colon \forall \q \in 
\triangle_\ab, 
\exists {W^n, \cT, 
}~\text{s.t.}~\eqref{eqn:cases-testing} \text{ holds}\},
\end{align*}
the smallest $\dist$ for which we can test if a distribution 
is $\dist$-far from $\q$. \textit{Uniformity testing (UT)} refers to the 
testing problem where we restrict $\q$ to be $\unif[\ab]$, the 
uniform distribution over $[\ab]$, we denote the corresponding risk as 
$\risk_{\dut}(\ab, \ns, \cW, \corr)$. 
\newer{We also denote the smallest $\alpha$ such 
that~\eqref{eqn:cases-testing} holds with probability success probability 
$\beta$ by $\risk_{\dit(\dut)}^\beta(\ab, \ns, \cW, \fracc)$.}

Without constraints and attacks, the risk is known to be
$\Theta(k^{1/4}/\sqrt{n})$~\cite{Paninski08, ChanDVV14}.

We now mention two special cases of our setting. 
\begin{itemize}
\item
When there are no information constraints (i.e., $\cW$ contains any 
scheme), 
users can transmit $Y_i=X_i$, namely the samples can be sent as is. Then 
a $\corr$-fraction of the samples are corrupted, reducing to the strong 
contamination model in
robust estimation where a $\corr$ fraction of the samples are 
corrupted. We denote the rates as $\risk_{\dl(\dit)}(\ab, \ns, \corr)$, \newer{dropping $\cW$ from the notation.} 
\item
When $\corr=0$, there is no manipulation and only 
information constraints are present, and we denote the rates by $\risk_{\dl(\dit)}(\ab, \ns, \cW)$. 
\end{itemize}
\medskip
\noindent\textbf{Organization.} \newer{We present our contributions and related work in Section~\ref{sec:contribution} and~\ref{sec:related} respectively. In  Section~\ref{sec:earth_mover}, we establish our lower bound technique based on earth-mover distance (EMD). In Section~\ref{sec:rit} we establish tight risk bounds without manipulation attacks ($\gamma=0$). In Section~\ref{sec:constrained} we show bounds for manipulation attacks under information constraints.}

\section{Our contributions} \label{sec:contribution}

\begin{table*}[hbt!]
	\begin{center}
		\begin{tabular}{|c|c|c|c|}
			\hline
			\multicolumn{1}{|c|}{Task} &\multicolumn{1}{|c|}{Constraint} & 
		 	 Manipulation risk (UB) & Manipulation risk (LB) \\
			\hline
			\multirow{6}{*}{DL}&\multirow{2}{*}{None}&  $O\Paren{\sqrt{\frac{\ab}{\ns}}+\fracc}$& $\Omega\Paren{\sqrt{\frac{\ab}{\ns}}+\fracc}$\\
			&& (Folklore)& (Folklore, Corollary~\ref{coro:rdl_main}) \\ 
			\cline{2-4}
			&\multirow{2}{*}{$\eps$-LDP}  & 
			$\tilde{O}\Paren{\sqrt{\frac{\ab^2}{\eps^2 n}}+ \frac{\sqrt{\ab} 
			}{\eps} \cdot \fracc} $ & $\Omega\Paren{\sqrt{\frac{\ab^2}{\eps^2 
			n}}+ \frac{\sqrt{\ab} }{\eps}\cdot \fracc}$ \small($\dagger$) \\
			&& \cite{cheu2019manipulation}& (Theorem~\ref{thm:lower_ldp}) 
			\\ 
			\cline{2-4}
			&\multirow{2}{*}{$\ell$-bit}  &$O\Paren{\sqrt{\frac{\ab^2}{2^\ell \ns}}+ \sqrt{\frac{\ab}{2^\ell}}\cdot \fracc  }$& $\Omega\Paren{ \sqrt{\frac{\ab^2}{2^\ell \ns}}+ \sqrt{\frac{\ab}{2^\ell }} \cdot \fracc }$ \\
			& &(Theorem~\ref{thm:lbit_learning}) & 
			(Theorem~\ref{thm:lbit_learning}) \\ \hline
			\multirow{6}{*}{IT}& \multirow{2}{*}{None}& 
			 $O\Paren{ \frac{\ab^{\frac14}}{\sqrt \ns}+ \gamma+\sqrt{ \frac{\ab\gamma}{\ns}}+ \sqrt[4]{\frac{\ab \gamma^2}{\ns}}}$ & $\Omega\Paren{ \frac{\ab^{\frac14}}{\sqrt \ns}+ \gamma+\sqrt{ 
			\frac{\ab\gamma}{\ns}}+ \sqrt[4]{\frac{\ab 
			\gamma^2}{\ns}}}$ \\
			&& (Theorem~\ref{thm:rit_main}) 
			&(Theorem~\ref{thm:rit_main}) \\ \cline{2-4}
			 &\multirow{2}{*}{$\eps$-LDP} & $O\Paren{\sqrt{\frac{\ab}{\eps^2 
			 n}}+\frac{\sqrt{\ab} }{\eps} \cdot \fracc}$ & 
			 $\Omega\Paren{\sqrt{\frac{\ab}{\eps^2 n}}+\frac{\sqrt{\ab} }{\eps} 
			 \cdot \fracc}$ \small($\dagger$)\\
			&& \cite{cheu2019manipulation}& (Theorem~\ref{thm:lower_ldp})\\ 
			\cline{2-4}
			&\multirow{2}{*}{$\ell$-bit} &   $O\Paren{ \sqrt{ \frac{\ab }{ 
			 2^{\ell/2} \ns }}+ \sqrt{\frac{\ab}{2^\ell}} \cdot 
			\Paren{ \gamma+\sqrt{ \frac{2^\ell\gamma}{\ns}}+\sqrt[4]{\frac{2^\ell
							\gamma^2}{\ns}}}}$ & $\Omega\Paren{ \sqrt{ \frac{\ab }{ 
			2^{\ell/2} \ns }}+ \sqrt{\frac{\ab}{2^\ell}} \cdot \Paren{\fracc+\sqrt{ \frac{2^\ell\gamma}{\ns}}} }$  \\
			&& (Theorem \ref{thm:lbit_testing}) & 
			(Theorem~\ref{thm:lbit_testing})  \\ \hline
		\end{tabular}
		\caption{Summary of results. For problems marked by {\small 
				($\dagger$)}, \cite{cheu2019manipulation} also 
		provides lower bounds lower than the stated
		bounds by a logarithmic 
		factor under a weaker threat model.}
		\label{tab:results_c}
	\end{center}
\end{table*}

\noindent\newer{\textbf{Lower bounds from EMD.}}
Since manipulation attacks can change a $\corr$-fraction of the $n$ messages, we characterize the difficulty of learning and testing under such attacks in 
terms of the earth-mover distance (EMD)
between messages with Hamming distance as the metric, \zsnew{stated in 
Theorem~\ref{thm:earth_mover}.}
	\zsnew{Using Le Cam's method, the lower bounds are provided in 
	terms of EMD between distributions of message 
	from mixtures of sources (distributions), which is critical for obtaining 
	tight bounds.} 

\smallskip
\noindent\textbf{Robust learning and testing.} 
\newer{Without information constraints, the server observes the 
true samples from $\p$ but with $\corr$-fraction adversarially corrupted.} For distribution learning the minimax risk is 
$\Theta(\sqrt{k/n}+\corr)$. While this result is standard, we provide it for completeness in Corollary~\ref{coro:rdl_main}. 
\newer{For testing, the optimal risk is more involved.} %
In Theorem~\ref{thm:rit_main} we show that when $\corr$ fraction 
of the samples are corrupted, the risk is 
$\Theta(\ab^{1/4}/\sqrt{\ns} + \gamma+\sqrt{ 
\ab\gamma/\ns}+\sqrt{\gamma} \sqrt[4]{\ab/\ns} 
)$ where the first term corresponds to the statistical rate proved in 
\cite{Paninski08, ValiantV14, DiakonikolasGPP18}.
\newer{In particular, when $\gamma\gg \min\{1/\sqrt{k}, 
1/\sqrt{\ns}\}$ the risk increases 
significantly compared to the uncorrupted case.}

\smallskip
\noindent\textbf{Manipulation attacks under information constraints.}
In Corollary~\ref{coro:lower_general}, we provide a general lower bound for 
estimating and 
testing distributions under information constraints and a $\corr$-fraction
manipulation attack. The result builds on the recently developed 
framework for distributed inference in~\cite{AcharyaCT19} and bounds 
the EMD between 
messages in terms of the trace norm of a channel information matrix 
(Definition~\ref{def:channel_matrix}).

\textit{Communication constraints.} \newer{In 
Theorem~\ref{thm:lbit_learning} and~\ref{thm:lbit_testing}, we establish 
risk bounds for
distribution learning and testing under $\ell$-bit communication constraints. We propose a protocol based on random hashing which matches the lower bound we prove up to logarithmic factors.} 
Our bounds 
suggest that manipulation attacks are significantly stronger with 
communication constraints on the channels. \zs{More precisely, the 
error due 
to manipulation attack can be as large as $\tilde{\Theta}(\fracc 
\sqrt{\ab/2^\ell})$ compared to $\fracc$ in the unconstrained setting.} We 
also 
provide a 
robust testing algorithm under communication constraints based 
on an  $\ell_1/\ell_1$ isometry 
in~\cite{acharya2020domain}. However, the bounds only match the lower 
bounds for 
$\ell = O(1)$ or $\ell = \Theta(\log \ab)$. The testing bound in the 
unconstrained case suggests more 
\newer{effort is} needed to study how communication constraints limit
EMD.
Closing this gap is an interesting future direction.

\textit{Privacy constraints.} \newer{In Theorem~\ref{thm:lower_ldp} we prove a lower 
bound that matches the upper bounds provided
in~\cite{cheu2019manipulation} for both testing (up to a constant factor) 
and learning (up to logarithmic 
factor).} 
We note that in~\cite{cheu2019manipulation}, a lower bound smaller by a logarithmic 
factor is proved under a 
weaker threat model, which is not directly comparable to our result.

The results are summarized in Table~\ref{tab:results_c}\footnote{All stated 
risk bounds are 
	upper bounded by $1$, which is omitted throughout the paper 
	for simplicity.}.

\section{Related work} \label{sec:related}
Without local information constraints, our work is related to the 
literature of robust statistical inference. Robust statistics has a long 
history~\cite{huber2004robust}. 
More recently, the interest focuses on designing computationally 
efficient robust algorithms in high-dimensional 
estimation~\cite{DiakonikolasKKLMS16,
	LaiRV16}.
See~\cite{diakonikolas2019recent} for a survey. 
In~\cite{DiakonikolasKKLMS16, ChenGR16}, it is proved that for estimating a single 
Gaussian distribution, the risk due to adversarial attack is 
$\Theta(\fracc)$. For discrete 
distributions, 
a line of 
work~\cite{QiaoV18, ChenLM20, 
	JainO20, JainO20-b} consider robust estimation in the 
distributed setting where each user contributes $s \gg 1$ samples. 
Compared to the result in Corollary~\ref{coro:rdl_main}, they show that 
in this case the risk due to manipulation can be much \newer{smaller} than $\fracc$.

\newer{\cite{ValiantV11a, DaskalakisKW18} study tolerant identity testing where the goal is to test between $\dtv{\p}{\q}\le \dist/10$ and $\dtv{\p}{\q}\ge \dist$ for a reference distribution $\q$.} The optimal sample complexity has been established as 
$\Theta(\ab/\dist^2\log\ab)$. Suppose $\fracc 
= \dist$, then a $\fracc$-robust 
identity tester is also a tolerant tester since with $\fracc$ fraction of the 
users controlled, the adversary can simply change the distribution to 
another distribution within TV distance $\frac{\dist}{2}$ with high probability.
Theorem~\ref{thm:rit_main} shows that the robust setting is strictly harder by showing that $\Theta (\ab)$ samples are needed when $\fracc$ and $\dist$ are both constants. This is due to the richer class of attacks that 
the adversary can perform compared to the tolerant testing where the 
samples are \newer{still independent.}

\newer{Robust identity testing Gaussian distributions without information constraints has been studied in~\cite{DiakonikolasK20}, where the contamination model is slightly different. In~\cite{AcharyaCT20}, identity testing of Gaussians is studied under communication constraints without manipulation attacks.
}

Without manipulation attacks, there is significant recent work interest in studying discrete 
distribution learning and testing in the distributed setting under 
information constraints. Optimal 
risks have been established under communication 
constraints~\cite{han2018isit, han2018geometric, acharya2020inference, AcharyaCT19, AcharyaCLST20} and LDP 
constraints~\cite{DuchiJW13, KairouzBR16, Sheffet17, AcharyaCFT18, AcharyaCLST20}. 

In distributed learning systems,  especially federated 
learning~\cite{kairouz2019advances}, manipulation attack is related to the 
so-called model poisoning attack~\cite{BCMC19, BVHES20}, where the 
attacker has full control of a 
fraction of the users and can change model updates arbitrarily which 
doesn't 
have to obey the local training and messaging protocol. In these works, it is 
shown that manipulation attacks can 
significantly outperform the classic data 
poisoning 
attack where the attacker can only insert data 
points to local users whose 
messages still follow the local messaging protocol.

\section{Moving the earth: the power of manipulation attacks} 
\label{sec:earth_mover}

\newer{We now characterize the power of manipulation attacks in terms 
of the earth-mover (a.k.a. Wasserstein) distance  between the distributions of the messages at the output of the channels. We 
first recall EMD with Hamming metric.}
\begin{definition}
	Let $\mdist_1$ and $\mdist_2$ be distributions over $\cY^n$ and 
	$\pi(\mdist_1, \mdist_2)$ be the set of all couplings 
	between $Q_1$ and $Q_2$. The 
	earth-mover distance (EMD) between $\mdist_1$ and 
	$\mdist_2$ is
	\[
	\dem{Q_1}{Q_2} := \!\!\!\! \inf_{Q \in \pi(Q_1, Q_2)} 
	\expectover{(\Yon_{(1)},\Yon_{(2)}) \sim Q}{\ham{\Yon_{(1)}}{\Yon_{(2)}}}.
	\]
\end{definition}

Note that a $\corr$-manipulation attack can change $Y_{(1)}^n\in\cY^n$ to another sequence $Y_{(2)}^n\in\cY^n$ as long as 
$\ham{Y_{(1)}^n}{Y_{(2)}^n} \le \fracc \ns$. If $\mdist_1$ and $\mdist_2$ are distributions over \newer{length-$\ns$} messages in $\cY^n$ 
with EMD at most $c\cdot\fracc n$, for some small constant $c$, 
the attack can effectively confuse sequences generated from $\mdist_1$ 
and $\mdist_2$. We formalize this intuition in 
Theorem~\ref{thm:earth_mover} below. A key ingredient in the theorem below is to consider {message distributions} from a \emph{mixture of input distributions}.

\newer{Let $q\in\setab$ be some distribution and $\cP \subset \triangle_{\ab}$ be a \emph{finite} set such that for all $\p \in \cP$,
\begin{align}
\label{eqn:large-tv}
	\dtv{\p}{\q}\ge \dist.
\end{align}
Let $\rp$ be uniformly drawn from $\cP$. Further, for a fixed $W^n\in\cW^n$ 
\begin{align}
\label{eqn:mixture}
\expectation{\rp^{\Yon}} = \frac1{|\cP|}\Paren{\sum_{p\in\cP} \p^{\Yon}}.
\end{align}
be the message distribution when the input distribution is uniformly chosen from $\cP$. 
}

\begin{theorem} \label{thm:earth_mover}
\newer{Suppose $\cP$ satisfies~\eqref{eqn:large-tv} for some $q$, and for all $W^n\in\cW^n$,
\[
\dem{\expectation{\rp^{\Yon}}}{\q^{\Yon}} \le
\frac{{\fracc \ns}}{2}.
\]}
Then for both distribution learning and testing,
	\begin{align*}
		\risk_{\dit}(\ab, \ns, \cW, \fracc) = \Omega\Paren{\dist},  &&
		\risk_{\dl}(\ab, \ns, \cW, \fracc) = \Omega\Paren{\dist}.
	\end{align*}
\end{theorem}

\begin{proof}
\hzz{We first show a reduction from testing to learning.}
Suppose there exists a distribution \newer{learning algorithm} with risk $\alpha/20$ under any
$\corr$-manipulation attack. We can use this for testing as follows: (i) By 
Markov's inequality, we learn the distribution to output a $\hp$ such 
that \newer{with probability} at least 0.9, $\dtv{\hp}{p} \le \alpha/2$ ,  
and then (ii) test if $\dtv{\hp}{q}\gtrless \alpha/2$ to perform testing with 
respect to $\q$. This shows that 
$\risk_{\dl}(\ab, \ns, \cW, \fracc)\ge c\cdot  \risk_{\dit}(\ab, \ns, \cW, \fracc)$,
for some $c$. Therefore, we only need to prove that $\risk_{\dit}(\ab, \ns, \cW, \fracc) = \Omega\Paren{\dist}$.

Fix $W^n \in \cW^n$, we have 
$\dem{\expectation{\rp^{\Yon}}}{\q^{\Yon}} \le{\fracc 
\ns}/{2}.$\footnote{\zsnew{Without loss of generality, we 
		assume $W^n$ is fixed since the adversary can observe public 
		randomness 
		$U$.}}
By the existence of minimizer for optimal 
transport\footnote{Hamming distance is a lower semi-continuous 
cost function.}, 
there 
exists a 
randomized
mapping $F$ such 
that if $\Yon_{(1)} \sim \expectation{\rp^{\Yon}}$, then \newer{$\Yon_{(2)} \stackrel{D}{=} F(\Yon_{(1)}) 
\sim \q^{Y^n}$} and $
	\expectation{\ham{\Yon_1}{F(\Yon_{(1)})}} \le \frac{{\fracc \ns}}{2}$.
By Markov's inequality, we have
\begin{equation} \label{eqn:prob_hamlarge}
	\probof{\ham{\Yon_{(1)}}{F(\Yon_{(1)})} >  \fracc \ns} \le \frac{1}{2}.
\end{equation}
Consider the following $\corr$-manipulation attack $M_C$:
\[
	M_C(\Yon_{(1)}) = \begin{cases}
	F(\Yon_{(1)}), & \text{ if } \ham{\Yon_{(1)}}{F(\Yon_{(1)})} \le \fracc \ns, \\
	\Yon_{(1)}, & \text{ if } \ham{\Yon_{(1)}}{F(\Yon_{(1)})} >  \fracc \ns.
	\end{cases}
\]
Then by~\eqref{eqn:prob_hamlarge}, we have
\[
	\dtv{M_C(\Yon_{(1)})}{\Yon_{(2)}} = \dtv{M_C(\Yon_{(1)})}{F(\Yon_{(1)})} \le \frac12.
\]

Hence, if the attacker sends the true messages when the distribution is 
$\q$, by Bayes risk, it is impossible to test between $\q$ and 
$\expectation{\rp^{\Yon}}$ with success probability at least 9/10. \hzz{By applying 
Le Cam's two-point method, it is impossible to tell whether the unknown distribution $\p$ equals $\q$, or comes from $\cP$, which concludes the proof.}
\end{proof}

With the main technique at hand, for each family of channels, we prove 
lower bounds by designing distributions that 
are separated in TV distances while the corresponding messages are close 
in earth-mover distance. We start with the unconstrained setting in 
Section~\ref{sec:rit} and turn to the constrained case in 
Section~\ref{sec:constrained}.

\section{Robust identity testing and learning} \label{sec:rit}
We consider robust identity testing without information 
constraints, i.e., the server observes the raw samples, $\corr$ fraction of 
which are adversarially corrupted. We prove the following tight minimax 
rate.
\begin{theorem} \label{thm:rit_main}
	\[
	\risk_{\dit}(\ab, \ns, \fracc) = 
	\Theta \Paren{ \frac{k^{1/4}}{\sqrt{\ns}} + \gamma+ 
		\sqrt{\frac{\ab \gamma}{\ns} } + \sqrt[4]{\frac{\ab 
				\gamma^2}{\ns}} }.
	\]
\end{theorem}

The first term is the statistical rate which is implied by the sample 
complexity bound in~\cite{Paninski08}. Our bound implies 
that when 
$\fracc > \min \{ 1/\sqrt{\ns}, 1/\sqrt{\ab}\}$, the risk due to 
manipulation can be significantly larger than the statistical risk. The upper 
bound is based on the $\ell_1$-tester proposed 
in~\cite{DiakonikolasGPP18}, which we present in 
Section~\ref{sec:rit_upper}. The lower bound is proved using the technique 
based on earth-mover distance
developed in Section~\ref{sec:earth_mover}, provided in 
Section~\ref{sec:rit_lower}.

\noindent We get the next corollary for learning under $\gamma$-manipulation attacks without information constraints.
\begin{corollary}[folklore] \label{coro:rdl_main}
	\[
	\risk_{\dl}(\ab, \ns, \fracc) = 
	\Theta \Paren{ \sqrt{\frac{\ab}{\ns}} + \gamma}.
	\]
\end{corollary}

The upper bound is achieved by the empirical distribution. The first term is the standard risk without manipulation, and the second term follows from the second term in Theorem~\ref{thm:rit_main} and the reduction from testing to learning. We omit the details.

\subsection{Upper bound for testing} \label{sec:rit_upper}

Our upper bound proceeds in two stages. We will first reduce identity 
testing to uniformity testing, and then provide an algorithm for uniformity 
testing.

\textbf{Reduction from identity to uniformity testing}.
\label{sec:ritreduction}
\newer{For unconstrained distribution estimation, Theorem 1 of~\cite{Goldreich16} showed that, up to constant factors, the risk for testing identity of any distribution is upper bounded by the risk of uniformity testing.} We extend their argument to the $\corr$-manipulation attack. In particular, we will show that
\begin{lemma}
	$$\risk_{\dit}(\ab, \ns, \fracc) \le 3\risk_{\dut}(6\ab, \ns, \fracc). $$
\end{lemma}

\begin{proof}
\cite{Goldreich16} showed that for any distribution $\q$ over $[\ab]$, there exists a randomized function $G_{\q}:[\ab]\to[6\ab]$  such that if $X\sim\q$, then 
$G_{\q}(X)\sim \uniff{6\ab}$, and if $X\sim\p$ for a distribution with 
$\dtv{\p}{\q}\ge\dist$, then $\dtv{G_{\q}(X)}{u[6\ab]}\ge\dist/3$.

Let $X^n$ be $n$ samples from $\p$, and $Z^n$ be a $\corr$-corrupted 
version 
of $X^n$. We then apply $G_\q$ independently to each of the $\per_\iter$ 
to 
    obtain a new sequence $G_\q(\perron)\eqdef( G_\q(\per_1)\ldots 
    G_{\q}(\per_\ns))$. If $\p = \q$, 
    $G_\q(X^n)$ is distributed according to $\uniff{6\ab}$ and 
    $G_\q(Z^n)$ is a $\fracc$-manipulated version of it. \zsnew{If 
    $\dtv{\p}{\q}\ge\dist$, $G_\q(X^n)$ is distributed according to a 
    distribution at least $\alpha/3$-far from $\uniff{6\ab}$.}
    Therefore, an algorithm for $\fracc$-robust uniformity testing with 
    $\dist' =\dist/3$ can be used for $\alpha$-identity testing with $\q$.
\end{proof}

\ignore{
Now we introduce our algorithm for $\gamma$-robust identity testing.
\begin{algorithm}
    \caption{$\gamma$-robust identity testing}
    \label{algorithm_identity}
    \begin{algorithmic}[1]
     \STATE {\bfseries Input:}  $\ns$ $\gamma$-corrupted samples 
     $\perron$, $\q$, $\dist$, 
    \STATE  Apply $F_\q$ independently to each of the $\per_\iter$ to 
    obtain a new sequence $F_\q(\perron)\ed F_\q(\per_1)\ldots F_{\q}(\per_\ns)$.
    \STATE Run $\gamma$-robust uniformity testing algorithm Algorithm~\ref{algorithm_uniformity}, with input 
    $F_\q(\perron)$ and $\dist = \dist/3$.
\end{algorithmic}
\end{algorithm}
Given $\ns$ samples $\outron$ from a distribution $\p$ over $[\ab]$, and $\perron$ are the $\gamma$-corrupted messages.
Note that  after the mapping, $F_\q(\perron)$ are still 
$\gamma$-corrupted, and for the unperturbed samples, $\per_\iter = 
F_\q(\out_\iter)$. Therefore, all the good samples in $F_{\q} (\perron)$ are 
generated according to distribution $F_{\q}(\out)$, which is $\uniff{6\ab}$ 
when $\out \sim \q$, or has a total variation distance of at least $\dist/3$ 
from $\uniff{6\ab}$ when $\out\sim\p$ with $\dtv{\p}{\q}\ge\dist$. 
Therefore, any valid $\gamma$-robust uniformity testing algorithm leads 
to a valid identity testing algorithm.
}

We \newer{now} present an algorithm for $\corr$-robust uniformity testing. 
Recall that $Z^\ns$ is obtained upon perturbing $\gamma n$ samples from $X^n\sim \p$. 
For $\perr \in [\ab]$, let $\smbb{\perron}$ be the number of 
appearances of $\perr$ in 
$\perron$. The empirical TV distance of $Z^n$ to uniform, which was used in~\cite{DiakonikolasGPP18},
\begin{align}
\label{equ:stat_rit}
\stat{\perron} \eqdef \frac12 \cdot \sum_{x=1}^{\ab} \absv{ \frac  
{\smbb{\perron}} {\ns} -\frac1{\ab} }
\end{align}
will be used as our test statistic. 

\newer{We now bound the difference in test statistic between $\Zon$ and $\Xon$ when $\ham{\Zon}{\Xon}\le \corr n$.
\begin{lemma} \label{lem:test_perturb} Suppose $\perron$ is obtained from  $\Xon$ by manipulating at most $\corr n$ samples, then
	\begin{align}
	\left| S(\Xon) - S(\perron)\right| \le \min\Paren{\gamma, 
		\frac{\ns\gamma}{\ab}}.\nonumber
\end{align}
\end{lemma}}
\begin{proof}
\newer{For the first term, by the triangle inequality for any $a,b\in \RR$, we 
have $||a|-|b||\le |a-b|$. Using this in~\eqref{equ:stat_rit}, we obtain,} 	
		\begin{align}
		\left| S(\Xon) - S(\perron)\right| &\le  \frac12\cdot 
		\sum_{x=1}^{\ab} 
		\absv{ \frac  {M_x({\Xon}) - M_x({\perron})  } {\ns} } \le \gamma, 
		\nonumber
		\end{align}
\newer{where we used the fact that changing \emph{one sample} from $\Xon$ changes $M_x(\Xon)$ 	for at most two $x\in[k]$ each by at most one, and therefore $\sum_{x=1}^{\ab} 
		\absv{M_x({\Xon}) - M_x({\perron})  }\le 2\corr n$}.

\newer{We note that the second term is smaller than the first when $\ns <\ab$. When $\ns<\ab$,} 
		\begin{align*}
		S(\Xon) &= \frac12\cdot \Paren{\sum_{x:M_x({\perron}) \ge 1} \Paren{\frac{\smbb{\Xon}}{\ns} -\frac1{\ab}} + \sum_{x:M_x({\Xon}) =0} \frac1{\ab}}\\
		&= \frac1{2}\Paren{1-\Paren{1-\frac{\Phi_0 (\Xon)}{\ab}} +\frac{\Phi_0 
		(\Xon)}{\ab} } \\
		&= \frac{\Phi_0 (\Xon)}{\ab}.
		\end{align*}
		where $\Phi_0 (\Xon)$ is the number of symbols not appearing in 
		$\Xon$.
		Therefore,
\[
		\left| S(\Xon) - S(\perron)\right| \le \frac1{\ab}\cdot |\Phi_0 (\Xon) 
		- \Phi_0 (\perron)| \le \frac{\ns\gamma}{\ab}. \qedhere
\]
\end{proof}

\newer{Next we use the following result from~\cite{DiakonikolasGPP18}, 
which shows a separation in the test statistic under $p=u[k]$ and 
$\dtv{\p}{\uni} \ge \dist$.} Let $\mu(\p)\eqdef \expectover{\Xon 
\sim 
p}{\stat{\Xon}}$ be the 
expectation of the statistic of the original samples. 
\begin{lemma}[\cite{DiakonikolasGPP18},~Lemma 4] 
\label{lem:uniformity_testing}
Let $\Xon$ be i.i.d. samples from $\p$ 
over $[\ab]$. For every $\beta \in (0,1),$ there exist constants $c_1, c_2$ 
such 
that if $\dist \ge 
c_1\cdot \frac{\ab^{\frac1{4}}}{\sqrt{\ns}}$, then with probability at least 
$1 - \beta$,
\begin{enumerate}
\item when $\p = \uni$,
\begin{center}
	$S(\Xon)- \mu (\uni) < \frac{9}{10} c_2 \dist ^2 \min \left\{ 
	\frac{\ns^2}{k^2}, \sqrt { \frac{\ns}{k} }, \frac1\dist\right\}$, \text{ and}
\end{center}
\item when $\dtv{\p}{\uni} \ge \alpha$,
\begin{center}
	$S(\Xon)- \mu (\uni) > \frac{11}{10}c_2 \dist ^2 \min \left\{ 
	\frac{\ns^2}{k^2}, \sqrt { \frac{\ns}{k} }, \frac1\dist\right\}$.
\end{center}
\end{enumerate}
\end{lemma}
Our test for uniformity is the following: 
\begin{equation}\label{eqn:test}
\cT(\Zon)=
\begin{cases}
\yes& \text{ if } {S(\perron)- \mu (\uni) \le c_2 \dist ^2 \min \left\{ 
		\frac{\ns^2}{k^2}, \sqrt { \frac{\ns}{k} }, \frac1\dist\right\}}\\
\no& \text{ otherwise.}	
\end{cases}
	\end{equation}
By Lemma~\ref{lem:test_perturb}, when $\dist \ge \frac{10}{c_2}\cdot 
\Paren{\gamma+ 
	\sqrt{\frac{\ab\gamma}{\ns}} + \sqrt[4]{\frac{\ab 
			\gamma^2}{\ns}}}$,
\begin{equation}
\left| S(\Xon) - S(\perron)\right|  \le \frac{1}{10}\cdot c_2  \dist ^2 \min \left\{ \frac{\ns^2}{k^2}, 
\sqrt { \frac{\ns}{\ab} }, \frac1\dist\right\}.\nonumber
\end{equation}
\newer{Setting $\beta = 1/10$ in Lemma~\ref{lem:uniformity_testing} shows that the  
test in~\eqref{eqn:test} solves the uniformity testing problem.}
\smallskip

\begin{remark} \label{rmk:arbitrary_probability}
	Note that the proof shows that for any constant 
	failure probability $\beta$, the risk for robust identity testing is 
the same as that in Theorem~\ref{thm:rit_main} up to a constant factor, i.e., 
for any $\beta$, there is a 
	constant $c(\beta)$, such that
	\[
		\risk_{\dit}^\beta(\ab, \ns, \cW, \fracc) \le 
		c(\beta) \Paren{\frac{\ab^{\frac14}}{\sqrt{\ns}}+ \gamma+\sqrt{ 
				\frac{\ab\gamma}{\ns}}+\sqrt[4]{\frac{\ab 
					\gamma^2}{\ns}} }.
	\]
	This will be important when we consider error boosting in the proof of 
	Theorem~\ref{thm:testuppermain}.
\end{remark}

\subsection{Lower bound} \label{sec:rit_lower}

\newer{In uniformity testing we have $\q = \unif[\ab]$.} We will use $\cP$ to be the following class of $2^{\ab/2}$ distributions from~\cite{Paninski08} indexed by $\vecz \in 
	\bcube$, i.e.,
	\begin{align} \label{eqn:paninski}
		\cP = \{\p_{\vecz}:\vecz\in\{\pm1\}^{\ab/2}\}\text{, where }\p_{\vecz}(2i-1) = \frac{1+\vecz_i\cdot 2\dist}{\ab}, \text{~} 
		\p_{\vecz}(2i) = \frac{1-\vecz_i\cdot 2\dist}{\ab}.
	\end{align}
Note that for any $\vecz \in \bcube$, 
	$\dtv{\p_{\vecz}}{\unif[\ab]} = \alpha$. The following lemma, proved 
	in~\cite{AcharyaSZ18} characterizes the earth-mover distance between 
	$\expectover{Z \sim \unif \bcube}{\p_Z^{n}}$ and 
	$\unif[\ab]^{n}$.
	\begin{lemma}[\cite{AcharyaSZ18} Lemma 7]
		\label{lem:coupling-hamming}
			\begin{align*}
		\dem{\expectover{Z \sim \unif 
				\bcube}{\p_Z^{\ns}}}{\unif[\ab]^{\ns}} 
			= O \Paren{\ns \cdot \min 
			\left\{ 
			\frac{\ns\dist^2}{\ab}, \frac{\sqrt{\ns} \dist^2}{\sqrt{\ab}}, \dist 
			\right \}}.
	\end{align*}
	\end{lemma}
Plugging in Theorem~\ref{thm:earth_mover}, we get the lower bound by setting the EMD to be $\fracc 
\ns/2$ and solving for $\alpha$.

\section{Robust constrained inference} 
\label{sec:constrained}
\newer{In this section we consider learning and testing under communication and LDP constraints. We first state the results, then in Section~\ref{sec:learning} and~\ref{sec:it_upper_lbit} establish the upper bounds and finally close with lower bounds in Section~\ref{sec:constrained_lower}.}

\medskip
\noindent\textit{Communication constraints.}
For distribution learning under $\ell$-bit communication constraints, we establish the following risk bound, which is optimal up to logarithmic factors. 
\begin{theorem} \label{thm:lbit_learning}
	\[
	\risk_{\dl}(\ab, \ns, \cW_\ell, \fracc) = 
	\tilde{\Theta}\Paren{\sqrt{\frac{\ab^2}{\ns( 2^\ell \land \ab)}} + \fracc 
	\sqrt{\frac{\ab}{2^\ell \land \ab}}
		}.
	\]
\end{theorem}
\newer{The first term is the risk under communication constraints without manipulation attacks~\cite{han2018isit, acharya2020inference}. The 
second term shows that if $\ell < \log \ab$, manipulation attack 
increases the 
risk by $\Theta\Paren{\fracc \sqrt{\frac{\ab}{2^\ell}}}$ compared to the 
the increase of $\Theta(\fracc)$ in the unconstrained setting (Corollary~\ref{coro:rdl_main}).}

For identity testing, we obtain the following risk bounds. 
\begin{theorem} \label{thm:lbit_testing}
\hz{Suppose $\ell  \le \log \ab$, }
	\begin{align*}
	\risk_{\dit}(\ab, \ns, \cW_\ell, \fracc) = 
	O\Paren{ \sqrt{ \frac{\ab }{ 
				2^{\ell/2} \ns }}+ \sqrt{\frac{\ab}{2^\ell}} \cdot 
		\Paren{ \gamma+\sqrt{ \frac{2^\ell\gamma}{\ns}}+\sqrt[4]{\frac{2^\ell
					\gamma^2}{\ns}}}},
	\end{align*}
and
 \[
	\risk_{\dit}(\ab, \ns, \cW_\ell, \fracc) = \Omega\Paren{  \sqrt{ \frac{\ab 
	}{ 
				2^{\ell/2} \ns }}+ \sqrt{\frac{\ab}{2^\ell}} \cdot \Paren{\fracc + \sqrt{ \frac{2^\ell\gamma}{\ns}}} }.
\]
\end{theorem}

The first term above is the risk of testing under $\cW_\ell$ without information constraints. 
\zsnew{The risk due to
manipulation attack in the upper bound is increased by a factor of 
$\sqrt{\frac{\ab}{2^\ell}}$ compared to the unconstrained case in 
Theorem~\ref{thm:rit_main}
.}
\newer{We remark that the upper and lower bounds above match (up to constant factors) for $\ell = O(1)$ and for $\ell = O(\log\ab)$ (when $\ell = \log k$, it matches the risk for unconstrained testing 
in Theorem~\ref{thm:rit_main}). We believe that the upper bound is tight and proving a better lower bound is an interesting future work.}

\medskip
\noindent\textit{Local privacy constraints.} We establish the following lower bounds for learning and testing under $\priv$-LDP. 

\begin{theorem} \newer{Suppose $\priv=O(1)$,}
\label{thm:lower_ldp}
	\begin{align}
		\risk_{\dl}(\ab, \ns, \cW_\priv, \fracc) = 
		\Omega\Paren{ 
			\sqrt{\frac{\ab^2}{\ns 
					\priv^2}}+\fracc \sqrt{\frac{\ab}{\priv^2}} }, \nonumber\\
		\risk_{\dit }(\ab, \ns, \cW_\priv, \fracc) = 
		\Omega\Paren{ \sqrt{\frac{\ab}{\ns 
					\priv^2}}+\fracc  \sqrt{\frac{\ab}{\priv^2}}}.\nonumber
	\end{align}
\end{theorem}

\newer{We note that~\cite{cheu2019manipulation} designs algorithms that achieve the bounds above up to a constant factor for testing and up to logarithmic factors for learning. However, their lower bounds are a logarithmic factor smaller than Theorem~\ref{thm:lower_ldp} under a weaker threat model, making the bounds incomparable. }

\subsection{Distribution estimation with $\ell$ bits under manipulation}
\label{sec:learning}
\newer{Without loss of generality, we assume $\ell<\log k$. We now present a scheme based on random hasing that achieves the upper bound for learning with the rate in Theorem~\ref{thm:lbit_learning}.}
Random hashing has been previously used for sparse estimation under communication constraints~\cite{acharya2021estimating}. 
\begin{definition}
	A random mapping $h: [\ab] \rightarrow [2^\ell]$ is a {random hashing function} if $\forall x \in [\ab], y \in [2^\ell]$, $\probof{h(x) = y} = \frac{1}{2^\ell}.$
\end{definition}
\newer{Let $\cT$ be the set of all $\ab \times 2^\ell$ binary matrices that have exactly one `1' in each row. 
A random hashing function $h$ is equivalent to a $T_h$ drawn uniformly at random from $\cT$ with the correspondence
\[
	T_h(x, y) = \idc{h(x) = y}.
\]
Now for any fixed $y\in[2^\ell]$, the $y$th column of $T_h$ 
has each entry as an independent $\bernoulli{1/2^\ell}$ random variable.}

\noindent We describe the protocol and the estimator below.
\begin{enumerate}
	\item \newer{Using randomness $U$ users obtain independent random hashing functions $h_1, \ldots, h_n$ 
	and send \[Y_i = h_i(X_i).\]}
	\item Upon receiving the manipulated samples $Z^n\in[2^\ell]^n$, the server outputs 
	\[
		\hp(Z^n) = \frac{2^\ell}{n (2^\ell -1) } \Paren{\sum_{i = 1}^\ns T_{h_i} 
		(\cdot, Z_i) - \frac{\ns}{2^\ell}}.
	\]
\end{enumerate}
\medskip

Without manipulation attacks, when the server receives $\Yon$, it has been shown in~\cite{acharya2021estimating} that
	\[
	\expectation{\dtv{\hp(\Yon)}{\p}} = O\Paren{\sqrt{\frac{k^2}{2^\ell \ns}}}.
	\]
Note that in Theorem~\ref{thm:lbit_learning}, the second term becomes larger than one when $\gamma > c\cdot \sqrt{2^\ell/\ab}$. We therefore focus on 	$\fracc < \sqrt{2^\ell/\ab}$.
	\hz{By the triangle inequality, it}
	suffices to bound 
	$\expectation{\dtv{\hp(\Yon)}{\hp(\Zon)}}$. \newer{Let $C$ be the set of 
	samples that are manipulated, and hence $|C|\le \corr n = m$. 
		Therefore,}
	\begin{align}
		\dtv{\hp(\Yon)}{\hp(\Zon)}
		& = \frac{2^\ell}{2 n (2^\ell -1) } 
		\norm{\sum_{i \in C } 
		(T_{h_i} 
			(\cdot, Y_i) - T_{h_i} 
			(\cdot, Z_i))}{1}  \nonumber \\
		& \le \frac{1}{\ns} \max_{%
		|C'| = \numc} 
			\max_{y_i, 
			z_i \in  [2^\ell]}\norm{\sum_{i \in C' } 
			(T_{h_i} 
			(\cdot, y_i) - T_{h_i} 
			(\cdot, z_i))}{1} \label{eqn:max}\\
		& = \frac{1}{\ns}  \max_{%
			|C'| = \numc} \max_{y_i, 
			z_i \in  [2^\ell]}\max_{v \in \{\pm1\}^\ab} v^T \!\sum_{i \in C' } 
			(T_{h_i} 
			(\cdot, y_i) \! -\! T_{h_i} 
			(\cdot, z_i)),
		\label{eqn:dtv2max}
	\end{align}
	where~\eqref{eqn:max} follows by maximizing over $C$, and~\eqref{eqn:dtv2max} since for $ u \in \RR^\ab$, 
	$\norm{u}{1} = \max_{v \in \{\pm1\}^\ab} v^Tu$.
	
\newer{Recall that for  $\forall i \in [\ab]$ and $y_i,z_i \in \cY$, $T_{h_i} 
(\cdot, y_i)$ and $T_{h_i} (\cdot, z_i)$ are both $k$-dimensional binary vectors with 
each 
coordinate an independent $\bernoulli{1/2^\ell}$.} 
Then for any $ C' \subset [\ns]$, with 
$|C'| = 
\numc$, $x \in [\ab]$, $y_i, z_i\in[2^\ell]$, and $v \in \{\pm1\}^\ab$,
$v^T \sum_{i \in C' } T_{h_i} (\cdot, y_i) - T_{h_i} (\cdot, z_i) )$ is 
distributed as the 
difference 
    between 
two $\binomial{\ab\numc}{1/2^{\ell}}$ random variables\footnote{It is 
possible that these two binomials are correlated. However, the union bound 
used after this still holds.} since $h_i$'s are 
independently generated. 

Hence, by 
Chernoff bound (multiplicative form) and union bound, we have $\forall 
\eta \in 
(0,\sqrt{\ab \numc/2^\ell})$,
\begin{align}\label{eqn:chernoff}
\probof{\! v^T\! \sum_{i \in C' } (T_{h_i} (\cdot, y_i)\! \!- \!T_{h_i} 
(\cdot, z_i) ) 
	\!> \!2 \sqrt{\frac{\ab \numc}{2^\ell} 
	}\eta}\!\!\le\! 2e^{-\frac{\eta^2}{3}}.
\end{align}
Taking union bound over ${\ns \choose \numc}$ subsets $C'$,  
$(2^\ell)^\numc 
\times (2^\ell)^\numc$ possible choices of $\{y_i, z_i\}_{i \in C'}$, and 
$2^\ab$ choices of $v \in  \{\pm1\}^\ab$, by~\eqref{eqn:dtv2max} 
and~\eqref{eqn:chernoff}, we 
have 
\begin{align*}
	\probof{\! \dtv{\hp(\Yon)}{\hp(\Zon)} \! > \! \frac2\ns  \sqrt{\frac{\ab 
				\numc}{2^\ell} }\eta } \!\le \! \frac{2 {\ns \choose \numc}  
			(2^\ell)^{2\numc} 2^\ab}{ e^{ 
		\frac{\eta^2}{3}}}.
\end{align*}
\hz{
For $\eta = \sqrt{6 (\ab + 2 \numc \ell \log 2 
+\numc\log\ns+\log 
(2^\ell/\fracc^2\ab)) }$, we have\footnote{$\ell < \log \ab$ implies the 
choice of $\eta 
\in(0,\sqrt{\ab\numc/2^\ell})$.},
}
\begin{align*}
	\probof{\dtv{\hp(\Yon)}{\hp(\Zon)} > \frac2\ns \sqrt{\frac{\ab 
				\numc}{2^\ell} }\eta } \le\fracc \sqrt{\frac{\ab}{2^\ell}}.
\end{align*}
Hence using $\ns = \numc/\fracc$, we have 
		\begin{align*}
		\expectation{\dtv{\hp(\Yon)}{\hp(\Zon)} } \le  \frac2n \sqrt{\frac{\ab 
				\numc}{2^\ell} }\eta +   \fracc \sqrt{\frac{\ab}{2^\ell}}  
		= \tilde{O}\Paren{\fracc \sqrt{\frac{\ab}{2^\ell} }+ 
			\sqrt{\frac{ \ab^2}{2^\ell 
					\ns} }}.
		\end{align*}

\subsection{$\ell$-bit identity testing under manipulation}
\label{sec:it_upper_lbit}
\newer{We now establish the upper bound in Theorem~\ref{thm:lbit_testing}. We first reduce the problem to an unconstrained testing problem over a domain of size $[2^\ell]$ that can be represented using $\ell$ bits and then invoke our bounds from Theorem~\ref{thm:rit_main} for robust identity testing without information constraints.}

\newer{For the reduction, we use the domain compression technique proposed 
in~\cite{acharya2020domain} to compress the observed samples to a 
smaller domain of size $2^\ell$. Moreover, the protocol preserves the TV distance between any pair of distributions up to a factor of $\Omega(\sqrt{2^\ell/\ab})$ with a constant probability. More precisely, we will use the following lemma from~\cite{acharya2020domain}.}

\begin{lemma}[Theorem 3.2~\cite{acharya2020domain}]
	\label{lem:domain_compression}
For any $\ell< \log \ab$, there exists a mapping $\mapping: \{0,1\}^* \times [\ab] \rightarrow 
	[2^\ell]$ and 	universal constants $c_1$ and $c_2$ such that $\forall \p, \q \in \setab$ with $\dtv{\p}{\q} \ge 
	\dist$, we have
	\begin{equation}
	\probover{\Ue}{\dtv{ \mapping (\Ue, \p)}{ \mapping (\Ue, \q)} \ge c_1 
	\dist 
		\cdot \sqrt{\frac{2^\ell}{\ab}} }\ge c_2,\nonumber
	\end{equation}
	where $\Ue$ is a public random string and with a slight abuse of notation we denote by $\mapping (\Ue, \p)$  the 
	distribution of  $\mapping (\Ue, 
	\inp)$ when $\inp \sim \p$.
\end{lemma}

\newer{Using this lemma, our testing scheme is the following.}
\begin{enumerate}
	\item Divide users into $\nb = \ceil{\log_{1-c_2/2}(1/10)}$ 
	disjoint batches 
	of 
	equal 
	size, where $c_2$ is the constant in 
	Lemma~\ref{lem:domain_compression}. For each batch $\batch_j, j \in 
	[\nb]$, generate an independent 
	public random string $\Ue_j$.
	\item Each user $i \in \batch_j$ sends \newer{the $\ell$-bit message} $Y_i = \mapping(\Ue_j, X_i)$.
	\item \newer{For} $j \in [\nb]$, let $Z^{(B_j)}$ be the messages received from 
	users in $B_j$. 
	Perform the robust testing algorithm in Section~\ref{sec:rit_upper} with 
	alphabet size $2^\ell$ 
	and distance $\risk^{\beta}_{\dit}(2^\ell, 
	\ns/\nb, N \fracc)$ where $\beta = \min\{c_2/2, 1 - \sqrt[\nb]{9/10}\}$.
	\item If all tests output $\yes$, output $\yes$, else, output $\no$.
\end{enumerate}

We now analyze the algorithm. 

\newer{In the null case, when $\p = \q$, $\mapping (\Ue, \p) = \mapping (\Ue, 
\q)$, the test in each batch outputs $\yes$ with probability at least 
$1 - \beta$ (see Remark~\ref{rmk:arbitrary_probability}). Since the batches are disjoint, all tests output $\yes$ with probability at least 
$(1 - \beta)^\nb \ge 9/10$ since $\beta \le  1 - \sqrt[\nb]{9/10}$.}

\newer{Now in the alternate case, suppose that 
\[
\dtv{\p}{\q} \ge \frac1{c_1}\cdot\sqrt{\frac{\ab}{2^\ell}}\cdot\risk_{\dit}^\beta(2^\ell, \ns/\nb, \nb\fracc).
\] 
Then with probability at least $c_2$ over the 
randomness of $\Ue$, we have
\[
	\dtv{ \mapping (\Ue, \p)}{ \mapping (\Ue, \q)} \ge c_1 \cdot
		\dtv{\p}{\q} 
		\cdot \sqrt{\frac{2^\ell}{\ab}}  = \risk_{\dit}^{\beta}\Paren{2^\ell, 
		\frac{\ns}{\nb}, 
		\fracc\nb}. 
\]}
\newer{Conditioned on this event, by Theorem~\ref{thm:rit_main}, we have the test 
in this batch 
outputs $\no$ 
with probability at least $1 - c_2/2$ since $\beta < c_2/2$. By 
disjointness of the batches, and the union bound we have at least one of the tests output $\no$ 
with probability at least $1 - (1 - c_2/2)^\nb \ge 9/10$ by the choice of 
$\nb$.}

\newer{From Remark~\ref{rmk:arbitrary_probability} we know that $\risk^{\beta}_{\dit}(2^\ell, 
\ns/\nb, \nb \fracc) $ is at most larger than
$\risk_{\dit}(2^\ell, 
\ns,  \fracc) $ by a multiplicative constant, which gives the following reduction.} 
\begin{lemma}
	\label{thm:testuppermain}
	\begin{equation}
	\risk_{\dit}(\ab, \ns, \cW_\ell, \fracc) =
	O\Paren{  \sqrt{\frac{\ab}{
	2^\ell \land  \ab}}\cdot \risk_{\dit}(2^\ell \land \ab, \ns, \fracc)
	}.\nonumber
	\end{equation}
\end{lemma}

\newer{To obtain the upper bound in Theorem~\ref{thm:lbit_testing}, we only need to plug in the bound of robust testing without information constraints from Theorem~\ref{thm:rit_main}.}

\subsection{Lower bounds by $\chi^2$-contraction}
\label{sec:constrained_lower}
\newer{In order to establish the lower bounds, by Theorem~\ref{thm:earth_mover}, it is sufficient to establish upper bounds on the EMD of messages $\Yon$ induced by a suitably chosen mixture of distributions from that induced by the uniform distribution.}

In particular, we will relate the EMD under information constraints to the 
channel information matrices of the allowed channels,
which describes how the channel can exploit the structure of the set of 
distributions in~\eqref{eqn:paninski} to solve inference tasks.
\begin{definition}[Channel Information Matrix] 
	\label{def:channel_matrix}For a channel $W: [\ab] 
	\rightarrow \cY$, the channel information matrix of $W$, denoted by 
	$H(W)$, is a $(\ab/2) \times (\ab/2)$ matrix and $\forall i_1, i_2 \in 
	[\ab/2]$,
	\begin{align*}
	H(W)(i_1, i_2) \eqdef \sum_{y\in\cY} \!\!\frac{(W(y\!\mid \!2i_1\! 
		\!-\!1)\!-\!W(y\!\mid\!
		2i_1))(W(y\!\mid \!2i_2\!\!-\!1)\!\!-\!\!W(y\!\mid\! 
		2i_2))}{\sum_{x\in[\ab]} 
		W(y\mid
		x)}.
	\end{align*}
\end{definition}

We will establish the following upper bounds on the EMD under
local information constraints.

\begin{lemma}
For $\cP$ defined in~\eqref{eqn:paninski}, and $\rp$ be uniformly drawn 
from $\cP$, and for any $W^n\in\cW^n$, let $\expectation{\rp^{\Yon}} 
= \frac1{2^{k/2}}\Paren{\sum_{p\in\cP} \p^{\Yon}}$ be the message 
distribution under a uniform mixture and channels $W^n$. Then,
\begin{equation*}
\dem{\expectover{}{\p_Z^{\Yon}}}{\unif[\ab]^{\Yon}}
\le 
2 \ns \dist \sqrt{\frac{\max_{W \in \cW }\norm{H(W)}{*}}{k}}. 
\end{equation*}
\label{lem:emd-bound}
\end{lemma}
\vspace{-15pt}
We will use this lemma with Theorem~\ref{thm:earth_mover} to obtain the 
following lower bound for robust inference.

 \begin{corollary}\label{coro:lower_general}
 	\[
 	\risk_{\dit (\dl)}(\ab, \ns, \cW, \fracc) = 
 	\Omega\Paren{\fracc\sqrt{\frac{\ab}{\max_{W \in \cW }
 				\norm{H(W)}{*}}}},
 	\]
 	where $\norm{\cdot}{*}$ denotes the trace norm of a matrix.
 \end{corollary}
\zsnew{
\noindent Under privacy and communication constraints, we have 
$\max_{W 
\in 
	\cW_\priv} 
\norm{H(W)}{*}\le \priv^2$ and $\max_{W \in \cW_\ell} 
\norm{H(W)}{*}\le 2^\ell$, which are proved in~\cite{AcharyaCT19}. Using 
Corollary~\ref{coro:lower_general}, we obtain the corresponding terms in 
the lower bound part of Theorem~\ref{thm:lbit_learning}, the 
~\ref{thm:lbit_testing}, 
and~\ref{thm:lower_ldp}.
The first terms in these bounds are the lower bounds without manipulation 
attacks and are
proved in~\cite{AcharyaCT19, han2018isit, DuchiJW13} respectively. 

\begin{remark}
Lower bounds without 
information constraints are automatically lower bounds of the constrained 
inference. Therefore, we get the lower bound of 
$\Omega\Paren{\sqrt{\frac{\ab\fracc}{\ns}}}$ in 
Theorem~\ref{thm:lbit_testing} from Theorem~\ref{thm:rit_main}.
\end{remark}
 }
\hzz{
Now it is enough to prove Lemma~\ref{lem:emd-bound}.

 \begin{proof}
 We will use the same lower bound construction stated 
 in~\eqref{eqn:paninski}. 
	We bound the 
	earth-mover distance using the naive coupling for length-$\ns$ 
	independent sequences which is equal to
	the sum of
	TV distances on each entry. Then we can relate the TV distances to the 
	$\chi^2$-divergence bounds proved 
	in~\cite{AcharyaCT19}, stated below.
	\begin{lemma}[Theorem IV.11~\cite{AcharyaCT19}]\label{lem:chi2bound}
	Let $\norm{\cdot}{*}$ denote the trace norm of a matrix,
		\[	
			\expectover{Z \sim \unif\bcube}{\chisquare{W \cdot \p_Z}{W  \cdot  \unif[\ab]}} 
=
			\frac{8\dist^2}{\ab}\norm{H(W)}{*},
		\]
		where $\forall \q$, $W \cdot \q$ denotes the distribution of the 
		message $Y$ 
		given the input $X \sim \q$, and $\p_Z$ is defined in~\eqref{eqn:paninski}.
	\end{lemma}
Let $Z \sim \unif \bcube$. Then we have
		\begin{align*}
		 \;&\;\dem{\expectover{}{\p_Z^{\Yon}}}{\unif[\ab]^{\Yon}} \\ 
				\le 
				\;&\expectover{}{\dem{\p_Z^{\Yon}}{\unif[\ab]^{\Yon}}} && 
				\hspace*{-0.35cm}\text{(Convexity)} \\
		\le \;& \expectover{}{ \sum_{i = 1}^{\ns} \dtv{W_i \cdot\p_Z}{W_i \cdot\unif[\ab]}}&& \hspace*{-0.35cm}
			\text{(Naive coupling)} \\
		\le \;&\sum_{i = 1}^{\ns} \expectover{}{\sqrt{  \frac12 \chisquare{W_i \cdot\p_Z}{W_i\cdot \unif[\ab]}}} && \hspace*{-0.35cm}
		\text{(Pinkser's Inequality)}  \\
		\le\; &\sum_{i = 1}^{\ns} \sqrt{\expectover{} { \frac12 \chisquare{W_i \cdot\p_Z}{W_i \cdot\unif[\ab]}}} && \hspace*{-0.35cm}
			\text{(Concavity)}  \\
		\le\; & 2 \ns \dist \sqrt{\frac{\max_{W \in \cW }\norm{H(W)}{*}}{k}}. 
		&& \hspace*{-0.35cm}
		\text{(Lemma~\ref{lem:chi2bound})} \qedhere
	\end{align*}

\end{proof}
}

\newpage
\bibliographystyle{alpha}
\bibliography{masterref}

\end{document}